\documentclass[aps,twocolumn,epsfig,eqsecnum,showpacs]{revtex4}
\usepackage{amsthm}
\usepackage{amsmath}
\usepackage{amsfonts}
\usepackage{amssymb}
\usepackage{graphicx}
\newcommand\be{\begin{eqnarray}}
\newcommand\ee{\end{eqnarray}}
\newcommand\ba{\begin{array}}
\newcommand\ea{\end{array}}
\newcommand{\ket}[1]{|#1\rangle}
\newcommand{\bra}[1]{\langle#1|}

\newtheorem{theorem}{Theorem}
\newtheorem{lemma}{Lemma}

\theoremstyle{definition}

\def\r{\rangle}
\def\l{\langle}
\def\Tr{{\rm Tr}}

\def\cH{{\cal H}}
\def\cS{{\cal S}}

\def\cT{{\cal T}}

\def\cF{{\cal F}}

\def\cE{{\cal E}}

\def\cA{{\cal A}}

\begin{document}
\title{Repeatable quantum memory channels}
\author{Tom\'a\v s Ryb\'ar$^1$ and M\'ario Ziman$^{1,2}$}
\affiliation{$^{1}$Research Center for Quantum Information, Slovak Academy of Sciences, D\'ubravsk\'a cesta 9, 845 11 Bratislava, Slovakia \\
$^{2}$ Faculty of Informatics, Masaryk University, Botanick\'a 68a,
602 00 Brno, Czech Republic}
\begin{abstract}
Within the framework of quantum memory channels we introduce the notion
of repeatability of quantum channels. In particular, a quantum
channel is called repeatable if there exist a memory device implementing
the same channel on each individual input. We show that
random unitary channels can be implemented in a repeatable fashion,
whereas the nonunital channels cannot.
\end{abstract}
\maketitle


\section{Introduction}

A quantum channel is any transformation taking a state $\varrho$ 
of a quantum system as an input and transforming it to some 
state $\varrho^\prime$ on the output. Within the standard model 
\cite{davies,nielsen}
of quantum dynamics the channels are represented by completely
positive trace-preserving linear maps acting on the set of
(trace-class) Hilbert space operators $\cT(\cH)$. Let us note that
quantum states are represented by so-called density operators, i.e.
positive trace-class operators with unit trace. The physical picture
of quantum channels as the correct description of the evolutions of open 
quantum systems follows from the Stinespring theorem \cite{stinespring}.
According to
this theorem each quantum channel can be understood as
the unitary evolution of the isolated ``supersystem'' composed of the
system and its environment, where the unitary evolution is governed
by the Schr\"odinger equation (see Fig.\ref{fig:one}). In other words 
quantum channels are implemented by suitable {\it quantum devices}
consisting of intrinsic degrees of freedom (associated with the environment) 
and acting on the system via particular interactions between the system
and the environment.
 
The statistical nature of quantum physics requires that 
the experiments must be typically repeated large number 
of times in order to make some relevant conclusions about the
properties of quantum devices. A typical example is the problem 
of quantum channel tomography, in which the goal is to identify which 
channel the given device is implementing.  Any estimation procedure is based on
repeated use of the device. Considering the above model of the quantum 
device implementing some quantum channel we encounter the following problem. 
Although the concept of the quantum channel itself does not need
any particular specification of the environment properties, for the
repeated use of the same device the details about the environment can play
a significant role. In particular, let us consider the following (not entirely 
realistic) example. Consider an optical device ``storing'' a single photon 
in some polarization state $\xi$. After inserting another photon 
(in the polarization state $\varrho_1$) this device
will output the photon that was originally stored in it and the new photon 
will remain stored in the device. From the theory point of 
view the device implements the transformation
mapping the whole state space onto the state $\xi$, i.e. 
$\cS(\cH)\mapsto\xi$, where by $\cS(\cH)$ we denote the system's state space. 
However, using the same fiber once more we get the
transformation $\cS(\cH)\mapsto\varrho_1$, i.e. the channel action is
completely different (unless $\varrho_1=\xi$). The main aim of this paper 
is to analyze and characterize the situations, in which the 
device can be reused infinitely many times and still implementing the same 
quantum channel. As we shall see such reusable devices would be good 
for saving resources. Instead of infinite amount of resources, needed to
provide the channel transformation forever, finite resources 
would be sufficient. 

The paper is organized as follows: In the Section II we shall recall 
the basics of quantum memory channels, after that in Section III we shall 
define the problems of reusability of quantum channels and prove
the main theorems. In the last Section we shall discuss the derived
results.

\section{The effect of memory}

In the case when the subsequent
actions of the device are independent of the previous ones we say that 
the device implements a {\it memoryless channel}. If the output does not
depend on future inputs we say that the channel is {\it causal}. 
Let us note that memoryless channels are automatically causal. 
In what follows we shall assume that all physically relevant channels
are causal. Under such condition it was shown in the seminal paper
by Kretschmann and Werner \cite{kretschmann} that each causal memory 
channel can be understood as a sequence of collisions between the 
system and its environment playing the role of the memory.
In the last few years different aspects of quantum memory channels 
attracted researchers 
\cite{kretschmann,macchiavelo,giovanetti,mancini,cerf,plenio} 
and many interesting
results have been achieved concerning capacity, structure and physical
implementations for memory channels.

A personification of the Stinespring's theorem describing 
one usage of a device implementing a quantum channel 
is depicted in Fig.\ref{fig:one}.
According to this picture the device is consisting of some 
internal degrees of freedom
forming the effective environment affecting the system transferred through 
the channel. We shall refer to this internal degrees of freedom as to 
channel's memory associated with the Hilbert space $\cH_{\rm M}$.
The interaction between the system and the memory is described by a unitary
transformation 
$U:\cH_{\rm M}\otimes\cH_{\rm S}\to \cH_{\rm M}\otimes\cH_{\rm S}$, where
$\cH_{\rm S}$ is the system's Hilbert space. Assuming that the memory is 
initialized to state $\xi$ the channel reads
\be
\cE_\xi[\varrho]=\Tr_{\rm M}[U(\xi\otimes\varrho)U^\dagger]\, ,
\ee
where $\Tr_{\rm M}$ denotes the partial trace over the memory system.

\begin{figure}
\begin{center}
\includegraphics[width=2.65cm]{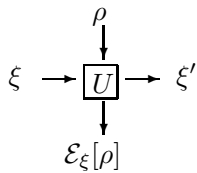}
\caption{One use of quantum channel \label{fig:one}}
\end{center}
\end{figure}

We see that due to the interaction the
state of the memory has changed. In particular
\be
\xi^{\prime}=\Tr_{\rm S}[U(\xi\otimes\varrho)U^{\dagger}]\equiv\cF_\varrho[\xi]\, .
\ee
So how to reuse the same device again and implement the same quantum channel? 
Clearly, the only way is to apply some \emph{reset} operation applied after 
each usage of the device and always initializing the memory into the 
fixed state $\xi$ (see Fig.\ref{fig:two}). Let us note that this memory state 
might be unknown for us. The reset operation can be achieved by using different 
procedures that are known as {\it relaxation processes}. 
However, let us note that the action of the reset transformation 
is a bit cheating, because it is not unitary and therefore, 
it can be implemented only by employing some additional
environment/memory system. Thus, resetting operation increases the total
cost of the channel's implementation measured in the size of needed
quantum resources. We shall see an example of the 
implementation of the reset operation at the end of Section III.
 
\begin{figure}
\begin{center}
\includegraphics[width=8cm]{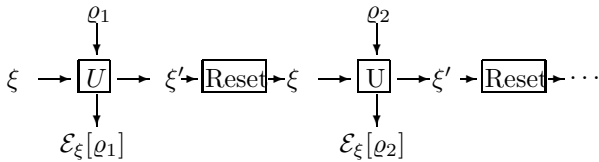}
\end{center}
\caption{Standard (memoryless) model.}
\label{fig:two}
\end{figure}

Nevertheless, within the framework of quantum (causal) memory channels  
\cite{kretschmann} the role of the reset operation is to suppress
the effect of memory on subsequent channel actions, i.e. 
to get device implementing a memoryless channel. In the problems 
dealing with channel estimation such reset procedures are implicitly 
assumed. In fact, for many systems the (approximate) relaxation 
is experimentally justified. Our goal here is not to analyze the 
effect of approximate relaxations, or verify the validity of this 
model in particular physical situations. We shall address more general 
question whether the perfect relaxation is possible, or needed, 
within the unitary model of quantum memory channel (see
Fig.\ref{fig:three}). In other words we are asking under which conditions 
the model depicted on Fig.\ref{fig:three} can embed the reset
model depicted in Fig.\ref{fig:two}.

\begin{figure}[!hbp]
\begin{center}
\includegraphics[width=6.5cm]{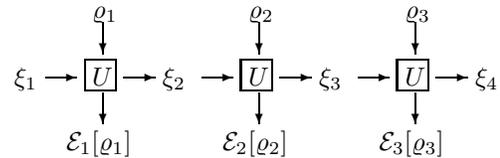}
\caption{Unitary memory model. \label{fig:three}}
\end{center}
\end{figure}

Let us now define the model and formulate the problem in an intuitive 
way. According to Fig.\ref{fig:three} after the $n$th usage of the device 
the memory is described by the state
\be
\nonumber
\xi_{n+1}&=&\Tr_{{\rm S}^{\otimes n}}
[U_n(\xi\otimes\varrho_1\otimes\cdots\otimes\varrho_n)U_n^\dagger)]\\
&=& \cF_n\circ\cdots\circ\cF_1[\xi_1]\, ,
\ee
where $U_n=U_{{\rm MS}_n}\cdots U_{{\rm MS}_1}$, $U_{{\rm MS}_j}$ acts
on the memory and $j$th input system, $\cF_j\equiv\cF_{\varrho_j}$ and
${\rm S}^{\otimes n}$ denotes the composite system of $n$ input systems.
In this model we assume that the input states $\varrho_1,\dots,\varrho_n$ 
are uncorrelated. 
On one hand this assumption is motivated by standard algorithms
for channel testing. On the other hand if we allow correlations, then
we are getting outside of the validity of the mathematical model for 
quantum channel as depicted in Fig.\ref{fig:one}. In fact, after 
the first usage of the device the subsequent inputs become 
to be correlated with the memory system. These issues extending the
standard quantum channel framework are studied in Refs.
\cite{pechukas,stelmachovic,jordan,zyczkowski,ziman}.

The induced channel transformation $\cE_j$ on the $j$th trial depends 
on the state of the memory $\xi_{j-1}$ which is dependent on the
particular choice of the input states $\varrho_1,\dots,\varrho_{j-1}$.
We shall investigate in which cases the particular choice 
of the input states does not matter. But before that, 
let us consider the following example.

\subsection{Memory channel induced by SWAP interaction}
Consider an experimentalist who would like to repeat his experiment 
aiming to describe the device schematically depicted on Fig.\ref{fig:one}.
If he is able to set the initial conditions of the experiment to some 
initial values and repeat the experiment he is fine. This refers to 
the model on Fig.\ref{fig:two}. If not, then he is not repeating the same 
experiment (with the same channel) again. This has some severe consequences 
for channel tomography. As an illustrative example we will consider 
a quantum device with a two-dimensional memory implementing a single qubit
channel. The interaction between the system (qubit) and the memory 
is described by the SWAP transformation $U_{\rm SWAP}$ acting as follows
\be
\varrho\otimes\xi\mapsto U_{\rm SWAP}(\varrho\otimes\xi)U_{\rm SWAP}^\dagger=
\xi\otimes\varrho\, .
\ee

In the memoryless settings this interaction induces completely contractive 
channels mapping the whole state space into the initial state of the memory
$\xi$. However, if the reset operation is not applied the situation is
completely different. Consequently, the result of the channel estimation
will depend on the particular algorithm we choose. Let us note that SWAP
transformation describes exactly the example mentioned in the introduction.
In particular, the $(j-1)$th input is mapped into $j$th output, i.e.
in $j$th run we observe the transformation $\varrho_j\mapsto\varrho_{j-1}$.

In order to estimate the qubit channel we can use six probe states:
the eigenvectors of $\sigma_x,\sigma_y,\sigma_z$ operators. In the usual
experiment we first insert $N$ times state $\psi_1$, and after that
$N$ times the state $\psi_2$, etc. In such case we shall observe
the transformations 
$\psi_1\mapsto \frac{1}{N}\xi+\frac{N-1}{N}\psi_1\approx\psi_1$ for 
large $N$ and similarly for any other test state. Based on this
we shall conclude that the channel is ideal, i.e. $\varrho\mapsto\varrho$.
However, starting to use this ideal channel for communication we very
quickly come into troubles. In the usual communication the states encoding
the characters are used randomly. If we performed the channel tomography
with randomly chosen test states we would find a completely different
channel. In fact, each state would be mapped into the complete mixture, since
the average test state is the total mixture. The differences between the 
ideal channel and completely noisy channel are obvious. As a result we see that
the standard channel tomography loses its point and new methods 
must be developed for the estimation of memory channels, but this 
problem is beyond the scope of this paper. 

\section{Repeatable quantum memory channels}

In order to avoid the problems mentioned in the previous Section we shall
focus on existence of reusable quantum devices implementing in each trial
the same channel. That is, our goal is to investigate the repeatability
of the quantum memory channel induced by a fixed unitary transformation 
$U$ as depicted in Fig.\ref{fig:three}. By {\it repeatable} quantum channels
we understand linear trace-preserving completely positive maps for which
there exists a unitary transformation $U$ and some initial
memory state $\xi$ such that $\cE_1=\dots=\cE_n=\cE$ for all $n>0$.
The key feature of repetable quantum memory channels is that 
the memory effects are suppressed. We call the triple $\l\cE,U,\xi\r$ 
a repeatable quantum memory channel of $\cE$. We have two basic questions:
\begin{itemize}
\item{}
Which channels are repeatable?
\item
Which channels are not repeatable?
\end{itemize}
Answering these two questions it is of interest to understand the consequences.
What does it mean when a channel is not repeatable? Does 
it mean that
it cannot be implemented at all? Let us remind that the considered model
is in fact the most general one in which the concept of channel makes 
sense. We shall get back to this question later at the end of this section.

A partial answer to the first question is given in the following theorem.
\begin{theorem}
If $\cE$ is a random unitary channel, i.e. 
$\cE[\varrho]=\sum_j p_j U_j\varrho U_j^\dagger$ and 
$0\le p_j\le 1,\sum_j p_j=1$, then it is repeatable.
\end{theorem}
\begin{proof}
We shall show that here exists a repeatable Stinespring dilation for any
random unitary channel. Consider a random unitary channel $\cE$. Define
a unitary transformation $U=\sum_j \ket{j}\bra{j}\otimes U_j$ 
acting on $\cH_{\rm M}\otimes\cH_{\rm S}$, with 
$\{\ket{j}\}$ being an orthonormal basis on $\cH_{\rm M}$
and $U_j$ are unitary transformation from the decomposition of the channel
$\cE$. The unitary transformations of such form are
also called {\it controlled-U transformations}. The
memory system plays the role of the controling system and 
the system itself is the target
system. Consider a general factorized input $\xi\otimes\varrho$
and calculate the states of the system and the memory after the
unitary transformation $U$ is applied. We obtain
\be
\nonumber
\cE[\rho] &=&  \Tr_{\rm M}[U(\xi\otimes\varrho)U^{\dagger}] \\
\nonumber &=&
\sum_{j,k} \Tr_{\rm M}[(\ket{j}\bra{j}\otimes U_i)(\xi\otimes\varrho)(\ket{k}\bra{k}\otimes U_k^{\dagger})] \\
\nonumber
&=& \sum_{j,k}\Tr[(\ket{j}\bra{j}\xi\ket{k}\bra{k})]\, (U_j\varrho U_k^{\dagger}) \\
&=&\sum_j \xi_{jj}\, U_j\varrho U_j^{\dagger}.
\ee
for the system's transformation and
\be
\nonumber
\cF[\xi] &=& \Tr_{\rm S}[U(\xi\otimes\varrho)U^{\dagger}] \\
\nonumber
&=& \sum_{j,k}\ket{j}\bra{j}\xi\ket{k}\bra{k}\,\Tr[U_j\varrho U_k^{\dagger}] \\
&=&\sum_{j,k}
\xi_{jk}\Tr[U_j\varrho U_k^{\dagger}]\,\ket{j}\bra{k}
\label{eq:ruc_mem}
\ee
for the memory transformation.

In order to implement random unitary channel $\cE[\varrho]=\sum_j p_j 
U_j\varrho U_j^\dagger$ it is sufficient to choose a state with diagonal 
elements $\xi_{jj}=\l j|\xi|j\r=p_j$. Let us note that diagonal elements
of density operator always form a probability distribution.
From Eq.(\ref{eq:ruc_mem}) it follows that diagonal elements of the
memory state are preserved, because
\be
\l j|\cF[\xi]|j\r=\sum_{j,k}\xi_{jk} \Tr[U_j\varrho U_k^\dagger]\, \delta_{jk}=
\xi_{jj}\, .
\ee
The last equality holds because $\Tr[U_j\varrho U_j^\dagger]=\Tr[\varrho]=1$.

Since the diagonal elements of $\xi$ defines the random unitary channels
and, moreover, they are preserved, it follows
that random unitary channels are indeed repeatable. In particular,
${\rm diag}[\xi_1]={\rm diag}[\xi_2]=\dots={\rm diag}[\xi_n]$
and therefore $\cE_1=\cE_2=\dots=\cE_n$ for all $n>0$.
\end{proof}

As a result we get that a particular Stinespring's dilation of any random unitary transformation forms a reusable quantum device, meaning that random unitary transformations are repeatable. Could it be that all transformations have such dilation? The following theorem gives a negative answer saying
that nonrepeatable channels do exist.

\begin{theorem}
If $\cE$ is a nonunital channel, i.e. $\cE[I]\ne I$, then it is not 
repeatable with finite memory.
\end{theorem}
\begin{proof}
We shall prove that repeatability of quantum memory channel (specified by $U$) 
implies unitality of the induced channels $\cE$. Let us start with the 
entropy analysis of the memory channel. Becuase of the unitarity it follows
that 
\be
S(\varrho_1)+S(\xi_1)=S(U(\varrho_1\otimes\xi_1)U^\dagger)\, ,
\ee
where $S(\varrho)=-\Tr[\varrho\log\varrho]$ is the 
von Neumann entropy of state $\varrho$. The entropy is subaditive,
i.e. $S(\omega_{AB})\le S(\omega_A)+S(\omega_B)$, where 
$\omega_A=\Tr_B\omega_{AB}$ and $\omega_B=\Tr_A\omega_{AB}$ are the
states of the subsytems $A,B$, respectively. Applying this inequality
for our situation we obtain
\be
S(\varrho_1)+S(\xi_1)\leq S(\cE[\varrho_1])+S(\xi_2)\, .
\ee
Let us repeat the quantum memory channel $n$ times by using 
the same input state, i.e. $\varrho_1=\varrho_2=\dots=\varrho_n=\varrho$. 
The repeatability of the channel $\cE$ implies that
\be
n S(\varrho_1)+S(\xi_1)\leq n S(\cE[\varrho_1])+S(\xi_{n+1})\, .
\ee
From this immediately follows the inequality
 \be \label{eq:ent}
 n\Delta (\varrho_1) \leq  S(\xi_{n+1})-S(\xi_1)\leq \log(\dim\cH_{\rm M}),
 \ee
where $\Delta (\varrho_1) = S(\rho_1)-S(\cE[\rho_1])$. 

For unital channels the entropy cannot decrease, i.e.
$\Delta (\varrho_1)\le 0$ for all states $\varrho_1$ (see
Appendix).
Consequently, the above
inequality is satisfied by all unital channels. For nonunital channels 
the complete mixture decreases its entropy, i.e. $\Delta (\frac{1}{d}I)>0$.
The right hand side is bounded by the dimension of the memory system.
However, since $n$ is arbitrarily large, the left hand side goes to infinity, 
hence necessarily also the dimension of the memory system must be infinite.
Thus repeatability requires unitality as it is stated in the theorem
\end{proof} 

We say that a channel is $n$-repeatable if its action can be repeated
$n$-times. For nonunital channels and finite memory there exist $n$ 
($n>\frac{\log\dim(\cH_{\rm M})}{\Delta_{\rm max}}$) such that the channel
cannot be $n$-repeatable. Let us now discuss the power of infinite 
memory systems. Consider a memory consisting of $n$ systems in the 
same state and of the same dimension as is the system
under consideration, i.e. $\xi_1=\xi^{\otimes n}$ 
is the initial state of the memory. The action of the memory channel
can be decomposed into a unitary interaction implementing
the desired channel $\cE$ (encoded in the state $\xi$) by acting only on
one subsystem of the memory and the input, and an operation permuting
the memory subsystems by one to the left, i.e. the active memory 
subsystem is shifted to the end. For such quantum memory channel 
the $j$th input is effectively interacting with the $j$th subsystem 
in the state $\xi$, hence each input is transformed by the same mapping 
$\cE_\xi[\varrho_j]=\Tr_{\rm M}[U(\xi\otimes\varrho_j)U^\dagger]$.
In this way arbitrary channel has $n$-repeatable implementation 
for all $n<\infty$. The limiting case $n\to\infty$ is physically senseless, 
because the Hilbert space of infinitely many subystems is not separable. 
On the other hand no one is probably interested in infinitely many 
repetitions of the same channel. Let us note that this type 
of implementation is essentially based on the complete replacement 
of the device by a new one with the same properties.

\section{Conclusion}
We investigated the problem of reusability of quantum devices
implementing (in each single use) state transformations described by 
quantum channels. Due to interaction of the system with 
the device both, the system and the device,
are affected by some noise, hence the original settings of the
device have changed. Consequently the
repeated usage of the same quantum device can result in a different noise,
i.e. different quantum channel. This picture leads to an emergence of memory 
effects in the description of quantum channels. If the channel can be
repeated infinitely many times without resetting the memory 
we say it is repeatable. For such type of channels the memory effects
are supressed although the memory itself undergoes a nontrivial dynamics. 
It was shown in this paper that any random 
unitary channel is repeatable with a finite memory, whereas the 
repeatable implementation of nonunital channels requires infinite resources. 
For qubit channels we can make even stronger statement that unitality
is equivalent to repeatability, because each unital channel can be expressed
as a random unitary channel \cite{ruskai}. For general systems we leave 
the question of repetable implementation of unital, but not random 
unitary channels open.

One possible way how to tackle the problem is to investigate the channels 
that can be implemented by a quantum device 
with the memory initialized in the total mixture. For such channels the
reset operation can be implemented in a repeatable way, since the channel
$\cA$ transforming the whole state space into the total mixture is random
unitary and therefore is repeatable. 
That is, whatever is the output memory state,
it can be reset to the total mixture by using only finite resources. 
Interestingly, since the entropy of the total system is preserved,
it follows that if the memory is initially in the total mixture, 
then the implemented channel is necessarily unital. In fact, 
if the system is initially in the total mixture, then necessarily 
also output must be in the total mixture, because the entropy achieves its
maximum for a unique state being the total mixture. But, this is nothing 
else as the unitality of the channel. It is an open problem whether there are
some unital but not random unitary channels that are implementable in the 
described repeatable way. 

Let us note that the concept of repeatability is similar to the
concept of quantum cloning \cite{gisin} in a sense that the channels 
(just like copies in quantum cloning) are not completely independent 
if measurements are taken into account. In fact, the memory system may act as 
a mediator of correlations between the channel outputs although the inputs
are factorized. For sure, the impact of measurements on repeatability 
of quantum memory channels deserves further investigation. 
The presented analysis of the repeatability of quantum channels is a part 
of the research program aiming to understand and develop realistic models 
of quantum dynamics of open systems including the memory effects. 

\acknowledgements This work was supported in part by the European
Union  projects QAP and the APVV project QIAM. TR acknowledges the support 
of the project APVV LPP-0264-07.
\begin{appendix}
\section{Monoticity of von Neumann entropy under unital channels}
\begin{lemma}
If $\cE$ is a unital channel, then $S(\cE[\varrho])\geq S(\varrho)$
for all states $\varrho$.
\end{lemma}
\begin{proof}
The proof of entropy monoticity for unital
channels is a consequence of the monoticity of the relative entropy
\cite{rel_entropy}. In particular, for arbitrary quantum channel $\cE$
\be
S(\cE[\varrho]||\cE[\omega])\leq S(\varrho||\omega)\,,
\ee
where $S(\varrho||\omega)=\Tr[\varrho(\log\varrho-\log\omega)]$ is the
quantum relative entropy. Setting $\omega=\frac{1}{d}I$ we get
$S(\varrho||I/d)=-S(\varrho)+\log d$. Using this fact and
assuming that $\cE$ is unital the above inequality can be rewritten as
\be
\nonumber
S(\cE[\varrho]||I/d)&\leq& S(\varrho||I/d)\\
\nonumber
-S(\cE[\varrho])&\leq& -S(\varrho)\, ,
\ee
from which the lemma follows.
\end{proof}

\end{appendix}



\end{document}